\documentclass[a4paper,10pt]{article} 
\usepackage[onehalfspacing]{setspace}%\usepackage[T1]{fontenc}
\usepackage[margin=2cm]{geometry}
\usepackage{authblk}
\usepackage{color}
\usepackage{enumitem}
\usepackage{amsmath,amssymb,amsthm,mathrsfs}
\usepackage{subcaption}
\usepackage{graphicx}
\graphicspath{{./Images/}}
\usepackage[noend]{algpseudocode}
\usepackage{algorithm}
\usepackage{color}
\usepackage{anysize}
\usepackage{multirow}
\usepackage{hyperref}
\usepackage{ragged2e}
\usepackage{multicol}
\newtheorem{theorem}{Theorem}

\newtheorem{lemma}{Lemma}

\newtheorem{obs}{Observation}

\title{Improved Total Domination and Total Roman Domination in Unit Disk Graphs}
\author{Sasmita Rout}
\author{Gautam K. Das}
\affil{Indian Institute of Technology Guwahati, India  781039\\ \{sasmita18,gkd\}@iitg.ac.in}

\date{}

\begin{document}
\thispagestyle{empty}
\maketitle
\begin{abstract}
Let $G=(V, E)$ be a simple undirected graph with no isolated vertex. A set $D_t\subseteq V$ is a total dominating set of $G$ if $(i)$ $D_t$ is a dominating set, and $(ii)$ the set $D_t$ induces a subgraph with no isolated vertex. The total dominating set of minimum cardinality is called the minimum total dominating set, and the size of the minimum total dominating set is called the total domination number ($\gamma_t(G)$). Given a graph $G$, the total dominating set (TDS) problem is to find a total dominating set of minimum cardinality. A Roman dominating function (RDF) on a graph $G$ is a function $f:V\rightarrow \{0,1,2\}$ such that each vertex $v\in V$ with $f(v)=0$ is adjacent to at least one vertex $u\in V$ with $f(u)=2$. A RDF $f$ of a graph $G$ is said to be a total Roman dominating function (TRDF)  if the induced subgraph of $V_1\cup V_2$ does not contain any isolated vertex, where $V_i=\{u\in V|f(u)=i\}$. Given a graph $G$, the total Roman dominating set (TRDS) problem is to minimize the weight, $W(f)=\sum_{u\in V} f(u)$, called the total Roman domination number ($\gamma_{tR}(G)$). In this paper, we are the first to show that the TRDS problem is NP-complete in unit disk graphs (UDGs). Furthermore, we propose a $7.17\operatorname{-}$ factor approximation algorithm for the TDS problem and a $6.03\operatorname{-}$ factor approximation algorithm for the TRDS problem in geometric unit disk graphs. The running time for both algorithms is notably bounded by $O(n\log{k})$, where $n$ represents the number of vertices in the given UDG and $k$ represents the size of the independent set in (i.e., $D$ and $V_2$ in TDS and TRDS problems, respectively) the given UDG.
\end{abstract}
\textbf{keywords: }{Dominating set, Total dominating set, Total Roman dominating set, Unit disk graphs, NP-complete, Approximation algorithms}
% \keywords------------- 
\section{Introduction}\label{Introduction}
Let $G$ be a simple undirected graph that may contain multiple components, but no component is an isolated vertex. In $G$, $V(G)$ and $E(G)$ represent the vertex and edge sets, respectively. $N_G(v)$ denotes the open neighborhood of $v$, and it is the set of all the neighbors (excluding itself) of $v$ in $G$, i.e., $N_G(v)=\{u\in V:uv\in E(G)\}$. Similarly, $N_G[v]$ denotes closed neighborhood  of $v$ and it is defined as $N_G[v]=N_G(v)\cup \{v\}$. $G[S]$ is the induced subgraph\footnote{for each $u,v\in S$, $uv\in E(G[S])$ if and only if $uv\in E(G)$} in $G$, where $S\subseteq V(G)$. The dominating set (DS) is an ordered partition of $V(G)$, say $(V_0,V_1)$, induced by a function $f:V(G)\rightarrow \{0,1\}$ such that $(i)$ $V_i=\{v\in V(G):f(v)=i\}$; $(ii)$ for each $v\in V_0$, there exists at least a vertex $u\in V_1$ such that $uv\in E(G)$. The dominating set with minimum $|V_1|$ is called the minimum dominating set, and the size of the minimum dominating set is called the domination number ($\gamma(G)$), i.e., $\gamma(G)=|V_1|$. In $G$, a vertex $v\in V(G)$ and a subset $S \subseteq V(G)$ dominate $N_G[v]$ and $\bigcup_{v \in S}N_G[v]$, respectively. A subset $D_t\subseteq V(G)$ is a total dominating set (TDS) if $(i)$ $D_t$ is a dominating set (\textit{domination} property), and $(ii)$ the induced subgraph $G[D_t]$ contains no isolated vertex (\textit{total} property). The total dominating set with minimum cardinality is called the minimum total dominating set, and the cardinality of the set is called the total domination number, $\gamma_t(G)$. Given a graph $G$, the objective of the minimum total dominating set problem is to find the total domination number of $G$.
The Roman dominating set (RDS) is an ordered partition of $V(G)$, say $(V_0,V_1,V_2)$ induced by a function, $f:V(G)\rightarrow \{0,1,2\}$ called Roman dominating function (RDF) such that \textit{(i)} $V_i=\{v\in V(G):f(v)=i\}$; \textit{(ii)} for each $v\in V_0$, there exists at least a vertex  $u\in V_2$ such that $uv\in E(G)$. The RDF with minimum weight, i.e., $W(f)=\sum_{u\in V(G)} f(u)$, is called the Roman domination number, $\gamma_R(G)$ and the corresponding ordered partition of $V$ is called the minimum Roman dominating set. The total Roman dominating set (TRDS) is an ordered partition of $V$, say $(V_0,V_1,V_2)$ induced by a function, $f:V(G)\rightarrow \{0,1,2\}$  called total Roman dominating function (TRDF) such that (i) $f$ is a Roman dominating function (\textit{Roman} property); (ii) the induced graph $G[V_1 \cup V_2]$ does not contain any isolated vertex (\textit{total} property). We define the total Roman domination number $\gamma_{tR}(G)$ as the minimum weight among all possible TRDFs on $G$. Given a graph $G$, the objective of the minimum total Roman dominating set problem is to find a total Roman dominating function of minimum weight.
\par Domination, a fundamental concept in graph theory, serves as a pivotal element across diverse fields such as computer networks, network security, telecommunication networks, and social networks. Unit disk graphs (UDGs) stand out as geometric representations where nodes correspond to points in the Euclidean plane, and edges connect nodes within a specified distance threshold. These graphs can be used as a prototype in modeling wireless communication networks, mobile Ad-Hoc networks (MANETs), and sensor networks. 
% In the context of an ad-hoc network where nodes possess varying security levels, we can model the network as a Roman Domination Problem within a unit disk graph (UDG). Here, each node in the network is associated with a security function $f:V(G)\rightarrow{{0,1,2}}$. This function dictates that nodes with security level $0$ are directly monitored by nodes with security level $2$, while nodes with security level $1$ self-monitor. Furthermore, nodes with nonzero security levels actively engage in the exchange of "still alive" messages within their neighborhoods.
The Roman Domination Set (RDS) problem aids in identifying optimal locations for enhancing security features within the network while minimizing costs. Simultaneously, the Total Roman Domination Set (TRDS) problem ensures that at least one monitoring node remains informed about any potential faults within the network. This approach contributes to maintaining network integrity and security effectively.
%-------------------------------------------------
\subsection{Related Work}\label{Related_Work}
 For a given graph $G$, the decision version of the dominating set decision problem is to find a dominating set of size at most $k$, where $k$ is a positive integer. In computational complexity theory, it is a classical NP-complete  problem \cite{hartmanis1982computers}. The domination problem and its variations are studied extensively in the literature \cite{haynes1998fundamentals,haynes2020topics,hedetniemi1991bibliography}. In $1990$, Clark et al. \cite{clark1990unit} showed that many classical problems, including domination, that are NP-complete for general graphs are also NP-complete for unit disk graphs. Subsequently, in 1994,  Marathe et al. \cite{marathe1995simple} gave a $5$-factor approximation algorithm for the dominating set problem in UDGs. In 1980, Cockayne et al. \cite{cockayne1980total}  introduced the concept of the total dominating set, and in $1983$, Pfaff et al. \cite{pfaff1983np} proved that the TDS problem is NP-complete for general and bipartite graphs. The TDS problem is polynomially solvable for trees, star graphs, complete graphs, cycles, and paths \cite{amos2012total,laskar1984algorithmic}. The detailed literature on TDS can be found in \cite{henning2009survey,henning2013total,sun1995upper}. Some more variations on total domination can be found in \cite{henning2009survey,henning2013total,rout2024semi,sun1995upper}.
In $2004$,  Cockayne et al. \cite{cockayne2004roman} introduced a new variation on domination called Roman domination. It was based on legion deployment with limited resources.  Shang et al. \cite{shang2010roman} introduced the concept of Roman dominating set (RDS) in UDGs and gave $5\operatorname{-}$factor approximation algorithm for Roman domination in UDGs. One of the variations, called total Roman domination, was introduced in \cite{liu2013roman}. In \cite{ahangar2016total}, authors established lower and upper bounds on the total Roman dominating set (TRDS). They related the total Roman domination number ($\gamma_{tR}(G)$) to domination parameters such as the domination number ($\gamma(G)$), Roman domination number ($\gamma_R(G)$) and total domination number ($\gamma_t(G)$). In \cite{martinez2020further}, authors gave a lower and an upper bound for $\gamma_{tR}(G)$, i.e.,  for any graph $G$ with neither isolated vertex nor components isomorphic to $K_2$, $\gamma_{t2}(G)+\gamma(G)\leq \gamma_{tR}(G)\leq \gamma_R(G)+\gamma(G)$ which was even tighter than the well known bound, $2\gamma(G)\leq \gamma_{tR}(G)\leq 3\gamma(G)$ in \cite{ahangar2016total}. In $2020$, A. Poureidi \cite{poureidi2020total} gave a linear time algorithm to compute the total Roman domination number for proper interval graphs. For general graphs, more results and variations on TRDS can be found in \cite{amjadi2017total,hao2020total,shao2019total}.
 \par In $2021$, Jena and Das \cite{jenatotal} demonstrated the NP-completeness of the Total Dominating Set (TDS) problem in Unit Disk Graphs (UDGs). They introduced an 8-factor approximation algorithm with a time complexity of $O(|V|\log{|D_t|})$, where $|D_t|$ denotes the size of the total dominating set. Additionally, the authors presented a Polynomial Time Approximation Scheme (PTAS) that operates in $O(k^2n^{2(\lceil{2\sqrt{2}}k\rceil)^2})$ time, capable of computing a total dominating set of size at most $(1+\frac{1}{k})^2\gamma_t(G)$, where $k\geq 1$. Despite this, the scheme achieves a total dominating set size of at most $4\gamma_t(G)$ in $O(n^{18})$ time, which is notably high. Moreover, as efforts for better approximation are made, the time complexity further escalates. Hence, there exists room for enhancing both the approximation factor and the running time efficiency of existing approaches.
\par Recently, in 2024, Rout et al. \cite{rout2024total} introduced the TRDS problem in UDGs and proposed a $10.5\operatorname{-}$factor approximation algorithm with running time $O(n\log{k})$. Additionally, in the same paper, they provided a $7.79\operatorname{-}$factor approximation algorithm for the TDS problem in UDGs, surpassing the approximation factor provided in \cite{jenatotal}, though with a slightly higher time complexity. In this paper, we have given
\subsection{Our Contribution}\label{Our_Contribution}
The subsequent sections of this paper are structured as follows. In Section~\ref{sec:Preliminaries}, we provide the necessary preliminaries and define the relevant notations. Following this, in Section~\ref{sec:NPC}, we establish the NP-completeness of the Total Roman Dominating Set (TRDS) problem for unit disk graphs. Sections~\ref{sec:aproxTotal} and \ref{sec:aproxTotalRoman} present our proposed approximation algorithms, namely TDS-UDG-SC and TRDF-UDG-SC, addressing the Total Dominating Set (TDS) and TRDS problems in UDGs. These algorithms offer approximation factors of $7.062$ and $5.92$ respectively. Lastly, in Section~\ref{sec:conclusion}, we conclude our findings.
\section{Preliminaries}\label{sec:Preliminaries}
\vspace{-.2cm}
This section begins by introducing necessary notations and redefining relevant definitions crucial to our exposition. Let $P=\{p_1,p_2,\dots,p_n\}$ represent a set of $n$ points in $\mathbb{R}^2$. A graph $G=(V,E)$ is said to be a geometric unit disk graph (UDG) corresponding to the point set $P$ if there exists a one-to-one correspondence between each $v_i\in V(G)$ and $p_i\in P$, and an edge $v_iv_j\in E(G)$ if and only if the Euclidean distance $d(p_i,p_j)\leq 1$. Here, $d(.,.)$ signifies the Euclidean distance between two points in $\mathbb{R}^2$. We denote the unit disk centered at point $p\in P$ as $\Delta(p)$, and collectively represent the set of disks $\Delta(P)=\{\Delta(p): p\in P\}$. The set of disks $\Delta(P)$ is considered independent if, for every pair $p,q\in P$, $p$ does not lie within the disk centered at $q$. Additionally, in this article, we interchangeably refer to points as vertices or nodes. For any subset $S\subseteq V$, $S$ is labeled independent if, for every pair $p,q\in S$, the distance between $p$ and $q$ exceeds $1$.

Furthermore, this section revisits some of the known lemmas and theorems concerning UDGs, Total Dominating Sets (TDS), and Total Roman Dominating Sets (TRDS). These established results are later applied in Sections $\ref{sec:NPC}$, $\ref{sec:aproxTotal}$, and $\ref{sec:aproxTotalRoman}$. Additionally, we highlight certain observations deemed pertinent for our proofs.
  \begin{lemma}\cite{marathe1995simple}\label{lem:5disks}
Consider a point $p\in \mathbb{R}^2$. If $S$ is the set of independent disks of radius $1$ such that each disk in $S$ contains the point $p$, then  $|S|\leq 5$.
\end{lemma}
 \begin{lemma}\cite{da2014efficient}\label{lem:sub-5}
Given the geometric representation of a UDG $G=(V,E)$, there exists a $\frac{44}{9}\operatorname{-}$approximation  algorithm for the minimum dominating set problem. 
\end{lemma}
% \begin{lemma}\cite{16}\label{lem:K2}
% If $f^*=(V_0^*, V_1^*, V_2^*)$ is an optimal TRDF on a graph $G=(V,E)$ such that $f^*$ attains minimum $|V_1^*|$, then either \textit{(i)} $V_2^*$ is a dominating set (DS) of $G$, or \textit{(ii)} $G[V_{11}^*]=\alpha K_2$ for some integer $\alpha \geq 1$, where $K_2$ represents the complete graph of two vertices.
% \end{lemma}
\begin{obs}\label{obs:minimum_V1}
   If $f=(V_0,V_1,V_2)$ is a TRDF on $G$, which attains minimum $|V_1|$, then the induced subgraph on $V_1$ and $V_2$ i.e., $G[V_1\cup V_2]$ does not contain a path\footnote{a sequence of non-repeated vertices connected through edges present in a graph} of three vertices (say, $x,y$ and $z$) with Roman value $1$ each (i.e., $f(x)=1$, $f(y)=1$ and $f(z)=1$).
\end{obs}
\begin{obs}\label{obs:D_T}\cite{goddard2014semitotal} 
   If $G$ is a graph with no isolated vertex, then $\gamma(G)\leq \gamma_t(G)$.
\end{obs}
\begin{theorem}\label{th:DomRom}\cite{ahangar2016total}
If $G$ is a graph with no isolated vertex, then $2\gamma(G)\leq \gamma_{tR}(G)$.
\end{theorem}
In this part of the section, we recall the set cover problem and an approximation algorithm for the set cover problem given in \cite{cormen2022introduction}. We will use the approximation algorithm as a subroutine in our algorithms in Sections~\ref{sec:aproxTotal} and \ref{sec:aproxTotalRoman}. We start by stating the set cover problem, followed by the result related to the approximation algorithm.\\
\textbf{Minimum Set Cover (MSC) problem:}\\
% If $U=\{u_1,u_2,\dots,u_n\}$ is the universal set and $S=\{S_1,S_2,\dots,S_m\}$ is the set of subsets such that $S_i\subseteq U$, $U=\bigcup\limits_{S_i\in S}{S_i}$, where $1\leq i \leq m$. Then the  MSC problem is to find a subset $T\subseteq S$ of minimum size such that $U=\bigcup\limits_{S_i\in T}{S_i}$. We denote the instance of MSC problem as $<U,S>$, where $U$ is a finite set, called a universal set, and $S$ is a family of subsets of $U$.
The Minimum Set Cover (MSC) problem involves a universal set $U=\{u_1,u_2,\dots,u_n\}$ and a set of subsets $S=\{S_1,S_2,\dots,S_m\}$, where each $S_i$ is a subset of $U$ and $U$ is the union of all $S_i$. In essence, the MSC problem seeks to identify a subset $T\subseteq S$ of minimal size such that $U$ is covered by the union of all subsets in $T$. We denote an instance of the MSC problem as $<U,S>$, where $U$ represents a finite set known as the universal set, and $S$ constitutes a family of subsets of $U$.
\begin{theorem}  \label{th:cormen}\cite{cormen2022introduction}
    The MSC problem can be approximated with an approximation factor $H(max\{|S_i|:S_i\in S  \})$ using GreedySetCover(U,S), where $H(m)$ is the $m\operatorname{-}{th}$ harmonic number.
\end{theorem}
\section{NP-completeness result}\label{sec:NPC}
In this section, our focus is on determining the complexity of the Total Roman Dominating Set (TRDS) problem. We aim to show that the decision version of the TRDS problem belongs to the NP-Complete class, especially when the graph is restricted to unit disk graphs. To accomplish this, we do a reduction from the decision version of the Dominating Set (DS) problem in grid graphs (GGs) to the decision version of the TRDS problem in unit disk graphs (UDGs). Formally, the respective decision versions are defined as follows:\\
\textbf{The decision version of the DS problem in grid graphs (D-DS-GGs)}: Given a nonzero positive integer $k$ and a grid graph $G$, does $G$ has a DS of size at most $k$?\\
\textbf{The decision version of the TRDS problem in UDGs (D-TRDS-UDGs)}: Given a nonzero positive integer $k'$ and a UDG $G$, does $G$ has a TRDS of weight at most $k'$?
\par In $1990$, Clark et al. \cite{clark1990unit} showed that \textbf{D-DS-GGs} belongs to the class NP-complete by doing a reduction from the known NP-complete problem \textbf{Planar Dominating Set of maximum degree $3$}. We prove the hardness result of \textbf{D-TRDS-UDGs} by making a polynomial time reduction from an arbitrary instance of \textbf{D-DS-GGs} to an instance of \textbf{D-TRDS-UDGs}.
\begin{lemma}\label{lem:DS_to_TRDS}
    If $G=(V,E)$ is an arbitrary instance of \textbf{D-DS-GGs} without any isolated vertex, then an instance $G'=(V',E')$ of \textbf{D-TRDS-UDGs} can be constructed from $G$ in polynomial time. 
\end{lemma}
\begin{proof}
    Let $G$ be a grid graph with grid size $1\times 1$, where $V(G)=\{p_1,p_2,\dots,p_n\}$ and $E(G)=\{e_1,e_2,\dots,e_m\}$ are the vertex set and edge set of $G$, respectively. We construct a UDG $G'=(V',E')$ from $G$ in two phases as explained below:
\par In the first phase, for each vertex $p_i\in V(G)$, we add a counterpart vertex $v_i\in V(G')$, where the distance and orientation between any two vertices $v_i,v_j\in V(G')$ are exactly the same as their counterpart vertices $p_i,p_j\in V(G)$. For each edge $p_ip_j\in E(G)$, we address its counterpart $v_iv_j$ as a segment in $G'$. Let $N$ be the set of points added in this phase ($|N|=|V(G)|=n$).
\par In the second phase, we add some auxiliary points in $G'$. For each segment $v_iv_j$ in $G'$, we add one point $x_{ij}$ on the grid line of the segment $v_iv_j$ at a distance $0.5$ either from $v_i$ or from $v_j$ (length of $v_iv_j$ is 1 unit). If $p_ip_j$ is parallel to $x$-axis and the point $x_{ij}$ is at the coordinate $(x,y)$, then we place another point $y_{ij}$ at coordinate $(x,y+0.1)$ (similarly, if $p_ip_j$ is parallel to $y$-axis and the point $x_{ij}$ is at the coordinate $(x,y)$, then we place the point $y_{ij}$ at coordinate $(x+0.1,y)$). We refer to each such $x_{ij}y_{ij}$ as a pendant edge. Let $A$ be the set of auxiliary points added in this phase to $G'$, i.e., $|A|=2|E(G)|=2m$.
\par Now, the graph $G'=(V',E')$ with vertex set $V'=N\cup A$ and edge set $E'=\{pq : p,q\in V'\text{ and } d(p,q)\leq 0.5\}$ is a unit disk graph.
 \par From the preceding steps it is evident that the UDG $G'$ can be generated from $G$  within polynomial time. Fig.~\ref{fig:GG-to-UDG}(a) and  Fig.~\ref{fig:GG-to-UDG}(b) provide a comprehensive illustration of the reduction using an example.
 \end{proof}
 \begin{figure}[h!]
\centering
\begin{subfigure}[b]{0.3\textwidth}
\centering
     \includegraphics[scale=.56]{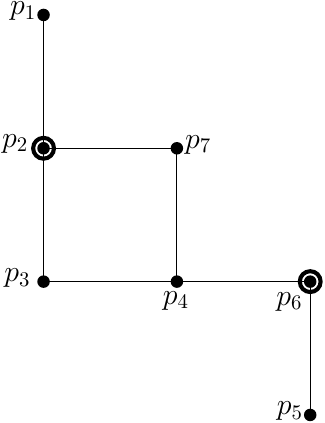} 
     \caption{$G=(V,E)$}
    \label{fig:GG-to-UDG-(a)}
\end{subfigure}
\begin{subfigure}[b]{0.3\textwidth}
\centering
    \includegraphics[scale=.56]{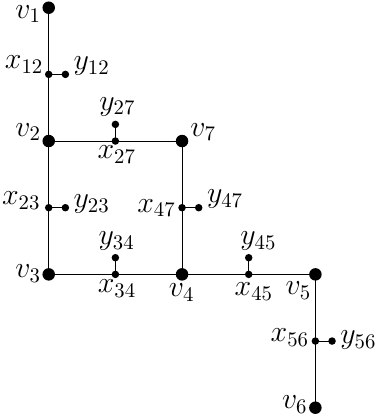} 
    \caption{$G'=(V',E')$}
    \label{fig:GG-to-UDG-(b)} 
\end{subfigure}
\begin{subfigure}[b]{0.3\textwidth}
\centering
    \includegraphics[scale=.56]{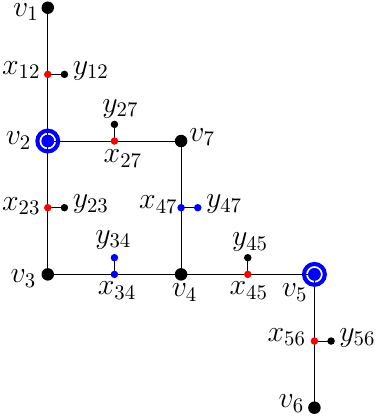} 
    \caption{$f=(V_0,V_1,V_2)$}
    \label{fig:UDG_12_disks2} 
\end{subfigure}
\caption{Illustration of $G$ and $G'$}
\label{fig:GG-to-UDG}
\end{figure}
 \begin{lemma}\label{lem:G'}
  Let $f=(V_0,V_1,V_2)$ be a TRDF of $G'$, such that $|V_1|$ attains minimum value. If $f(v_i)=1$, then for each $j$, $f(x_{ij})=2$, where $v_iv_j$ is a segment in $G'$. 
\end{lemma}
 \begin{proof}
   On the contrary, assume  there exists a segment $v_iv_j$ such that $f(v_i)=1$ and $f(x_{ij})=1$. If so, then $f(y_{ij})=1$. We can always have another TRDF $f'$ such that $f'(v_i)=1$, $f'(x_{ij})=2$ and for all other vertices $f(v)=f'(v)$.  This leads to a contradiction that $|V_1|$ attains minimum value in $f$.
   \end{proof}
  \begin{theorem}
     $\operatorname{D-TRDS-UDGs}$ is NP-complete.
 \end{theorem}
 \begin{proof}
     Given a positive integer $k$ and an ordered partition $f=(V_0,V_1,V_2)$ of $V(G)$, we can verify whether $f$ is a TRDF of $G$ and $W(f)\leq k$ or not in polynomial time. Hence, $\operatorname{D-TRDS-UDGs}\in NP$.
     \par To prove $\operatorname{D-TRDS-UDGs}$ is $\operatorname{NP-hard}$, we use a polynomial time reduction from  $\operatorname{D-DS-GGs}$ to $\operatorname{D-TRDS-UDGs}$. First, we construct an instance of $\operatorname{D-TRDS-UDGs}$ from an arbitrary instance of $\operatorname{D-DS-GGs}$ (i.e., from $G=(V,E)$ to $G'=(V',E')$) using Lemma~\ref{lem:DS_to_TRDS}. Fig.~\ref{fig:GG-to-UDG} gives a complete illustration of the reduction through an example. To complete the hardness result, we prove the following claim.
     \par \textbf{Claim: }$G$ has a dominating set $D$ of size at most $k$, if and only if $G'$ has a total Roman dominating function $f=(V_0,V_1,V_2)$ of weight at most $k+2m$.\\
     \textbf{Necessity:} To prove this, we consider the TRDF, which attains minimum $|V_1|$; otherwise, for a specific DS of $G$, there may exist multiple TRDFs for $G'$. Let $D=\{p_i:p_i\in V(G)\}$ be a dominating set of the grid graph $G$ such that $|D|\leq k$. 
     \par Now, we will assign the Roman value to each vertex $v\in V'$ as listed below:\\
     $(i)$ if $p_i\in D$, then $f(v_i)=1$, $(ii)$ if $p_ip_j\in E(G)$  and $p_i \text{(and/or } p_j\text{)} \in D$, then $f(x_{ij})=2$, $(iii)$ if $p_ip_j\in E(G)$ and $p_i,p_j\notin D$, then $f(x_{ij})=f(y_{ij})=1$, and $(iv)$ the remaining vertices in $G'$ carry Roman value $0$.

    Let $(V_0,V_1,V_2)$ be the ordered partition of $V'$ such that  $V_i=\{v\in V':f(v)=i\}$. Now, we argue that $f=(V_0,V_1,V_2)$  is a TRDF.
 For each edge $p_ip_j\in E(G)$, there exists a counterpart segment $v_iv_j$ in $G'$. Each segment $v_iv_j$ is associated with an edge $x_{ij}y_{ij}$, where $y_{ij}$ is a pendant vertex as shown in Fig.~\ref{fig:GG-to-UDG}(b). So to Roman dominate the vertices of the edge $x_{ij}y_{ij}$, the required Roman value is exactly $2$ (either $f(x_{ij})=2$ or $f(x_{ij})=f(y_{ij})=1$). If $p_i\in D$, then it dominates each $p_j$ in $G$, where $p_ip_j\in E(G)$. In contrast, in  $G'$, a Roman value $2$ to $x_{ij}$ ensures the Roman domination of $v_i$, $v_j$ and $y_{ij}$, and a Roman value $1$ to $v_i$ ensures the total Roman domination of $v_i$, $v_j$ and $y_{ij}$. In each edge $p_ip_j\in E(G)$, if $p_i,p_j\notin D$,  the vertices $p_i$ and $p_j$ are dominated by some other vertices, say $p_k$ and $p_{\ell}$, respectively. Therefore, the corresponding counterpart vertices $v_i,v_j$ in $G'$ are Roman dominated by the vertices $x_{ik}$ and $x_{j\ell}$ (i.e., $f(x_{ik})=2$ and $f(x_{j\ell})=2$) and can be total Roman dominated by assigning Roman value $1$ to $v_{k}$ and $v_{l}$ (i.e., $f(v_{k})=1$ and $f(v_{\ell})=1$). Hence, the edge $x_{ij}y_{ij}$ associated with the segment $v_iv_j$ needs a Roman value $2$ (i.e., $f(x_{ij})=1$ and $f(y_{ij})=1$) for total Roman domination. Therefore, each segment requires a Roman value of weight $2$ for Roman domination. Since there is $|E(G)|$ number of edges in $G$, which is equal to the number of segments in $G'$, Roman domination of $G'$ requires a weight of at least $2|E(G)|$, and total roman domination requires an additional weight with value $|D|$. Hence $W(f)=|D|+2|E(G)|\leq k+2m$, where $W(f)$ is the weight associated with the TRDF of $G'$.\\
 \textbf{Sufficiency: }Let $f=(V_0,V_1,V_2)$ be a TRDF of $G'$ of weight $W(f)\leq k+2m$. We prove that $G$ has a dominating set $D$ such that $|D|\leq k$.
 \par Given a TRDF $f=(V_0,V_1,V_2)$ of $G'$ with minimum $|V_1|$, we construct a set $D$ as follows: if $v_i\in V_1$, then $p_i$ is in  $D$ and if $x_{ij}\in V_2$ and $y_{ij}\in V_1$, then $p_j$ is in $D$.
 Now, we show that the set $D$ is a dominating set of $G$. If $f(v_i)=1$, then for each $j$, $f(x_{ij})=2$, where $v_iv_j$ is a segment in $G'$ (see Lemma~\ref{lem:G'}). Since $f(x_{ij})=2$ ensures the Roman domination of each $v_j$ in $G'$, inclusion of $p_i$ in $D$ ensures the domination of each $p_j$ in $G$, where $p_ip_j\in E$. Let $S_1$  be the set containing such $p_i$s. Next, we have to show that for each $v_i\in V'$, if $f(v_i)=0$, then there exists an edge $p_ip_k\in E$ such that $p_k\in D$.
Since $f$ is a TRDF, if $f(v_i)=0$, then in order to satisfy the \textit{Roman} property, there exists a point $x_{ik}\in V'$ such that $f(x_{ik})=2$. In addition to this, it must have satisfied the \textit{total} property, therefore, either $f(y_{ik})=1$ or $f(v_k)=1$. Whatever may be the case, as per the construction of $D$, $p_k\in D$.
\par In $G'$,  for any segment $v_iv_j$, if $f(v_i)=f(v_j)=0$, then $f(x_{ij})=2$ and $f(y_{ij})=1$, as $x_{ij}$ ensures the Roman domination of $v_i$, $v_j$ and $x_{ij}$. However, in $G$ to dominate $p_i$ and $p_j$, the inclusion of exactly one vertex from $p_i$ and $p_j$ in $D$ is sufficient.  So as per the construction of $D$, $p_j$  ensures the domination of $p_i$ and $p_j$, where $p_ip_j\in E(G)$. Let $S_2$ be the set of such $p_j$'s.
\par From the above two paragraphs, we conclude that 
$D=S_1\cup S_2$ is a dominating set of $G$. Now, we are left to show that $|D|\leq k$. 
\par Since each pendant edge corresponding to each segment in $G'$ requires at least a weight of value $2$ for Roman domination, the associated weight for total Roman domination in $G'$ is at least $2|E(G)|=2m$ (Since the number of segments in $G'$ is exactly equal to the number of edges in $G$). Since $W(f)\leq k+2m$, $|D|=|S_1\cup S_2|\leq k$. Therefore, $\operatorname{D-TRDS-UDGs}\in NP\operatorname{-}hard$.
\par Since $D\operatorname{-}TRDS\operatorname{-}UDGs\in NP$ and $D\operatorname{-}TRDS\operatorname{-}UDGs\in NP\operatorname{-}hard$, $D\operatorname{-}TRDS\operatorname{-}UDGs\in NP\operatorname{-}complete$.
\par To get a complete illustration of mapping from DS to TRDF, refer to Fig.~\ref{fig:GG-to-UDG}. If $D=\{p_2,p_5\}$ is a dominating set of the example graph $G$, then $f$ is the corresponding TRDF on $G'$. In Fig.~\ref{fig:GG-to-UDG}(a), the double circle black vertices represent the set $D$ of $G$, and in Fig.~\ref{fig:GG-to-UDG}(c), the red, blue, and black vertices of the graph $G'$ represent the sets $V_2$, $V_1$ and $V_0$ of the TRDF $f$, respectively.
\end{proof}
%-------------------------------------------------------------------------------
\section{A $7.17$ factor approximation algorithm for TDS problem} \label{sec:aproxTotal}
In this section, we propose a $7.17\operatorname{-}$ factor approximation algorithm for the TDS problem in geometric UDGs.
%----------------- sub-section 4.1 ------------------
\subsection{Algorithm}\label{subsec:algorithm}
Given a geometric unit disk graph $G=(V,E)$ with $V=\{p_1,p_2,\dots,p_n\}\subseteq \mathbb{R}^2$ as the set of disk centers, Algorithm~\ref{alg1:TDS_UDG_SC} finds a TDS $D_t$ of $G$. Now, we describe the procedure for finding the set $D_t$. First, the algorithm finds a maximal independent set $D\subseteq V$ of $G$ to satisfy the domination property (see Line~\ref{alg1:begin_loop1_find_dom}-\ref{alg1:end_loop1_find_dom} of Algorithm~\ref{alg1:TDS_UDG_SC}). Next, to satisfy the total property, the algorithm chooses a set of neighboring vertices $T\subseteq V$ such that for each $v\in D$, there exists a vertex $u\in V\setminus D$ and $u\in N_G(v)$. To find the set $T$, the algorithm constructs a set cover instance $<D,S>$ out of the sets $D$ and $V\setminus D$. If $V\setminus D=\{u_1,u_2,\dots, u_{|V\setminus D|}\}$, then for each $u_i\in V\setminus D$, it finds the neighbors of $u_i$ in $D$. Since $D$ is a maximal independent set of $G$, $|N_G(u_i)\cap D|\leq 5$ (due to Lemma~\ref{lem:5disks}). Let $S_i$ represents the corresponding neighbors of $u_i$ in $D$ and $S$ be the collection of such $S_i$, where $1\leq i\leq (|V\setminus D|)$ (see Line~\ref{alg1:begin_loop2_find_Si}-\ref{alg1:end_loop2_find_Si} of Algorithm~\ref{alg1:TDS_UDG_SC}). Now, the set $D$ is the universal set, $S=\{S_i:1\leq i\leq |V\setminus D|\}$ is the set of subsets, and the pair $<D,S>$ is an instance of the set cover problem. A complete illustration of the scenario is shown through an example in Fig.~\ref{fig:Figure1_2}(a), in which the edges between the set $D$ and $V\setminus D$ are only shown. Next, the algorithm calls a procedure called $GreedySetCover(D,S)$ to find the cover of the set $D$. Let 
$S'$ be the set returned by the procedure (see Line~\ref{alg1:setcover} of Algorithm~\ref{alg1:TDS_UDG_SC}). For each $S_i\in S'$, $T$ contains the corresponding $u_i$ of $V\setminus D$ (see Line~\ref{alg1:Set T} of Algorithm~\ref{alg1:TDS_UDG_SC}). Then it reports $D_t=D\cup T$ as a TDS of $G$. Lemma~$\ref{lem:correctness1}$  and Lemma~\ref{lem:TimeComplexity1} represent the algorithm's correctness and time complexity, respectively.
% ------------------------------------------------------------------
	\begin{algorithm}[h!]
	\small
		\noindent
		\textbf{Input: }{A unit disk graph, $G=(V,E)$, with known disk centers}  \\
		\textbf{Output: }{A TDS $D_t$ for $G$}
		\begin{algorithmic}[1] 
			\State $V^{\prime}=V,D=\emptyset,S=\emptyset$
			\While{$V^{\prime}\neq \emptyset$}  \Comment{\textit{domination} property of TDS}\label{alg1:begin_loop1_find_dom} 
		    \State choose a vertex $v\in V^{\prime}$ 
		    \State $D=D\cup \{v\}$ 
		    \State $V^{\prime}=V^{\prime}\setminus N_G[v]$ 
		    \EndWhile   \label{alg1:end_loop1_find_dom}
			\State $i=1$
			\For{each $u\in V\setminus D$}  \label{alg1:begin_loop2_find_Si} 
		    \State $S_i=N_G(u)\cap D$
		    \State $i=i+1$
                \State $S=S\cup S_i$
		    \EndFor \label{alg1:end_loop2_find_Si}
		    \State $S'=GreedySetCover(D,S)$ \label{alg1:setcover}
            \State $T=\{u_i|S_i\in S'\}$  \Comment{\textit{total} property of TDS}\label{alg1:Set T}
		    \State $D_t=D\cup T$ 
		    \State \Return $|D_t|$ \label{alg1:return_f}
			\caption{\textit{TDS-UDG-SC($G$)}}
			\label{alg1:TDS_UDG_SC}
		\end{algorithmic}
	\end{algorithm}
% 	----------------------------------------------------------------
\begin{figure}[h!]
\centering
\begin{subfigure}[b]{0.48\textwidth}
\centering
     \includegraphics[scale=.55]{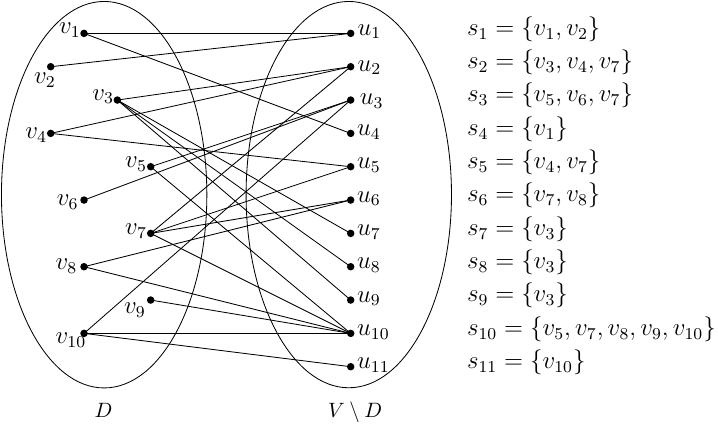} 
    \caption{}
    \label{fig:Figure1}
\end{subfigure}
\begin{subfigure}[b]{0.48\textwidth}
\centering
    \includegraphics[scale=.55]{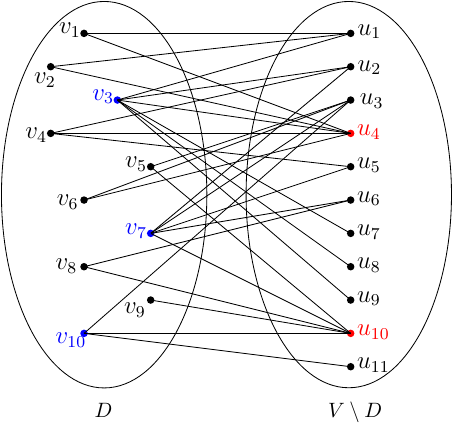} 
    \caption{}
    \label{fig:Figure2} 
\end{subfigure}
\caption{Illustration of $D$ and $V\setminus D$}
\label{fig:Figure1_2}
\end{figure}
\begin{lemma}\label{lem:D_VminusD}
Given a maximal independent set $D$ of a geometric unit disk graph $G=(V,E)$ such that $S=\{S_i\}$, where $S_i=D\cap N_G(u_i)$, for each $u_i\in V\setminus D$. If $C^*$ is an optimal set cover of the set cover instance $<D,S>$ and $D^*$ is an optimal dominating set of $G$, then $|C^*|\leq |D^*|$. 
\end{lemma}
\begin{proof}  
Since $D$ is a maximal independent set of the unit disk graph $G$, for any vertex $v\in V\setminus D$, $|N_G(v)\cap D|\leq 5$ (Due to Lemma~\ref{lem:5disks}). If $D^*$ is an optimal dominating set of $G$, then either $(i)$ $D^*\subseteq D$ or $(ii)$ $D^*\subseteq V\setminus D$ or $(iii)$ $D^*= D_D^*\cup D_{V\setminus D}^*$, where $D_D^*= D^*\cap D$ and $D_{V\setminus D}^*= D^*\cap (V\setminus D)$. Now, we have to show that for each of the cases $|C^*|\leq |D^*|$. 
\par $(i)$ $D^*\subseteq D$: if $D^*\subseteq D$, then $D^*=D$; otherwise, there exists at least one vertex $v\in D$, such that $v$ is not dominated by $D^*$ (since $D$ is an independent set). This leads to a contradiction that $D^*$ is a dominating set. Since $D$ is an independent set, in the worst case,  $D$ requires at most $|D|$ number of subsets (i.e., vertices) from $S$ (i.e., $V\setminus D$)  to cover $D$. If $C^*$ is an optimal cover of the set cover instance $<D,S>$, then  $|C^*|\leq |D|=|D^*|$.
\par $(ii)$ $D^*\subseteq V\setminus D$: On the contrary, let us assume $C^*> D^*$. Since $D^*$ is a dominating set, each vertex in $D\cup V \setminus D$ (i.e., $V$) is either in $D^*$ or is adjacent to at least one vertex in $D^*$. If so, a set $D'\subseteq D^*$ exists such that $D'$ dominates $D$ (since $D$ is an independent set). This implies $|D'|$ number of subsets (i.e., vertices) from $S$ (i.e., $V\setminus D$) are sufficient to cover $D$, which is less than or equal to $|C^*|$. This leads to a contradiction to the fact that $C^*$ is an optimal cover of $D$.
\par $(iii)$ $D^*= D_D^*\cup D_{V\setminus D}^*$, where $D_D^*= D^*\cap D$ and $D_{V\setminus D}^*= D^*\cap (V\setminus D)$ (refer to Fig.~\ref{fig:Figure1_2}(b) for a pictorial representation): Since $D_D^*\subseteq D$, $D_D^*$ is an independent set. In the worst case, at most $|D_{ D}^*|$  number of subsets (vertices) from $S$ are required to cover vertices in $D_D^*$. Since $D_{V\setminus D}^*$ dominates the set $D\setminus D_D^*$, $|D_{V\setminus D}^*|$ number of subsets (vertices) are sufficient for the coverage of remaining vertices in $D\setminus D_D^*$. Therefore,  $|C^*|\leq |D_D^*|+|D_{V\setminus D}^*|=|D^*|$.
\end{proof}
\begin{lemma}\label{lem:correctness1}
The set $D_t$ in \textit{TDS-UDG-SC} is a TDS of $G$.
\end{lemma}
\begin{proof}
In the first phase, \textit{TDS-UDG-SC} finds a maximal independent set $D$ of $G$ to satisfy the domination property (see line~ \ref{alg1:begin_loop1_find_dom}-\ref{alg1:end_loop1_find_dom} in Algorithm~\ref{alg1:TDS_UDG_SC}). Next, the algorithm runs $GreedySetCover(D,S)$ to find a subset $T$ such that $T=\{u_i|S_i\in S'\}$ (see line~\ref{alg1:setcover} and line~\ref{alg1:Set T} in Algorithm~\ref{alg1:TDS_UDG_SC}). The set $T$ ensures that for each vertex $v\in D$, there exists a vertex $u\in T$ such that $uv\in E(G)$. Hence, the set $T$, when combined with $D$, confirms that none of the vertices in $D$ is isolated. Therefore, combinedly the nominated points in $D$ and $T$ satisfy the \textit{domination} and \textit{total} properties.  Hence, the set $D_t$ is a TDS of $G$.
\end{proof}
\begin{lemma}\label{lem:TimeComplexity1}
 \textit{TDS-UDG-SC} runs in $O(n\log{k})$ time, where $n$ is the input size and $k$ is the size of the total dominating set.
\end{lemma}
\begin{proof}
Let $V=\{p_1,p_2,\dots,p_n\}$ be the set of disks' centers corresponding to graph $G=(V,E)$. Let all the disks lie on a rectangular region $\mathbb{R}$ of a plane. Let the rectangle's extreme left and bottom arms represent the $x$- and $y$-axis, respectively. Then we split the plane $\mathbb{R}$ so that the region $\mathbb{R}$ becomes a grid with cell size $1\times 1$. Let $[x,y]$ be the index associated with each cell, where $x,y\in \mathbb{N} \cup \{0\}$. If a point $p\in V$ is located at co-ordinate $(p_x,p_y)$ on $\mathbb{R}$, then the point belongs to a cell with index $[\lfloor{p_x}\rfloor,\lfloor{p_y}\rfloor]$.
\par In phase \textit{one}, \textit{TDS-UDG-SC} constructs a maximal independent dominating set $D$ of the input graph $G$. To do so efficiently, each non-empty cell maintains a list that keeps the points of $V$ chosen for inclusion in $D$ located within that cell. While considering a point $p\in V$ as a candidate for the set $D$, it only probes into $9$ cells surrounding the cell where $p$ lies. That means if  $p$ is located at co-ordinate $(p_x,p_y)$, then it searches in each $[i,j]$ cell, where  $\lfloor{p_x}\rfloor - 1\leqslant i \leqslant \lfloor{p_x}\rfloor + 1$ and $\lfloor{p_y}\rfloor - 1\leqslant j \leqslant \lfloor{p_y}\rfloor + 1$.\footnote{Any point outside these $9$ cells is independent from $p$} If there does not exist any point $q\in D$ in those $9$ cells such that $p\in \Delta(q)$, then $p$ is included in $D$. A height balance binary tree containing non-empty cells is used to store the points that are in $D$. Since each cell of size $1\times 1$ can contain at most $3$ independent unit disks (Since the perimeter of a cell is $4$ unit, the number of independent unit disks that a cell can contain is at most $3$; otherwise, the disks are no longer independent), the processing time to decide whether a point is in $D$ or not requires $O(\log k)$ time, where $k=|D|$. Thus the time taken to process $|V|=n$ points is $O(n\log{k})$ (see line~ \ref{alg1:begin_loop1_find_dom}-\ref{alg1:end_loop1_find_dom} in Algorithm~\ref{alg1:TDS_UDG_SC}). Each $S_i$ can be found by investigating the surrounding $9$ cells of $u_i$. This requires $O(\log{k})$ time. Hence, the total time for the steps  \ref{alg1:begin_loop2_find_Si} to \ref{alg1:end_loop2_find_Si} in Algorithm~\ref{alg1:TDS_UDG_SC} requires $O(n\log{k})$ time.
\par In phase \textit{two}, it uses $GreedySetCover()$ to find a set $S'$, which contains at least one neighboring vertex for each element present in $D$. Since $|V\setminus D|\leq n$ and $|S_i|\leq 5$, where $1\leq i\leq |V\setminus D|$ (Since $G$ is a unit disk graph and the set $D$ is a maximal independent set, due to Lemma~\ref{lem:5disks}, each vertex in $V\setminus D$ can have at most $5$ neighbors in $D$). The algorithm for $GreedySetCover()$ can be implemented in $O(n)$ time. 
\par Therefore, the running time of the algorithm \textit{TDS-UDG-SC} is $O(n\log{k})$.
\end{proof}
%----------------- sub-section 4.2 ------------------
\subsection{Analysis}
The set $D_t$ in \textit{TDS-UDG-SC} is a TDS of $G$, where $D_t=D\cup T$ (see Lemma~\ref{lem:correctness1}). Let $D^*$ and $D_t^*$ be an optimal DS and an optimal TDS of $G$, respectively. Since $D$ is a maximal independent set of $G$ and it is known from \cite{da2014efficient} that given the coordinates of each point in a graph $G$, the adjacency list of the graph can be found in polynomial time. Therefore, due to Lemma~$2$, the following claim is true.
\begin{equation}\label{eq-1}
  |D|\leq \frac{44}{9}|D^*|  
\end{equation}
The set $T$ in \textit{TDS-UDG-SC} satisfies the total property when added to the independent set $D$, where $T=\{u_i|S_i\in S'\}$ and  $S'$ is the set cover of the instance $<D,S>$. Let $C^*$ be an optimal set cover of $<D,S>$. Then from Lemma~\ref{lem:D_VminusD}, the following equation holds. 
\begin{equation} \label{eq-2}
    |C^*|\leq |D^*|
\end{equation} 
Therefore, from Lemma~\ref{lem:correctness1}, we conclude the following:
\begin{equation} \label{eq-3}
\begin{split}
|D_t|& =|D|+|T| \\
& \leq \frac{44}{9}|D^*|+H(5)\times |C^*| \text{ (refer to Equation~\ref{eq-1} and Theorem~\ref{th:cormen})}\\
& \leq \frac{43}{9}|D^*|+\frac{137}{60}\times |D^*| \text{ (refer to  Lemma~\ref{lem:D_VminusD})}\\
& = \frac{1291}{180} \times |D^*| \\
& \leq \frac{1291}{180} \times |D_t^*| \text{ (refer to  Observation~\ref{obs:D_T})}\\
& \approx 7.17 \times |D_t^*|
\end{split}
\end{equation}
\begin{theorem}
 The proposed algorithm (\textit{TDS-UDG-SC}) gives a $7.17\operatorname{-}$factor approximation result for the TDS problem in UDGs. The algorithm runs in $O(n\log{k})$ time, where $n$ is the input size and $k$ is the size of the independent set (i.e., $D$) in TDS problem.
\end{theorem}

\begin{proof}
    The approximation factor and the time complexity result follow from Equation~\ref{eq-3} and Lemma~\ref{lem:TimeComplexity1}, respectively.
\end{proof}
%-------------------------Section-------------------------------------------------
\section{A $6.03$ factor approximation algorithm for TRDS problem} \label{sec:aproxTotalRoman}
In this section, we propose a $6.03\operatorname{-}$ factor approximation algorithm for the TRDS problem in geometric UDGs.
%----------------- sub-section 5.1 ------------------
\subsection{Algorithm}
Now, we describe the procedure for finding a TRDF of a UDG $G=(V,E)$. First, we find a maximal independent set $V_2\subseteq V$ of $G$ to satisfy the Roman property. Next, to satisfy the total property, we choose a set of neighboring vertices $V_1\subseteq V$ such that for each $v\in V_2$, there exists at least one vertex $u\in V\setminus V_2$ and $u\in \Delta (v)$. The steps of finding the sets $V_2$ and $V_1$ in this algorithm are similar to those of finding the sets $D$ and $T$ in Algorithm~\ref{alg1:TDS_UDG_SC}, respectively. See Algorithm~\ref{alg2:TRDF_UDG_SC} (\textit{TRDF-UDG-SC}) for the complete pseudocode, Lemma~$\ref{lem:correctness2}$ for the correctness and Lemma~\ref{lem:TimeComplexity2} for the time complexity of the algorithm.
% ------------------------------------------------------------------
	\begin{algorithm}[h!]
	\small
		\noindent
		\textbf{Input: }{A unit disk graph, $G=(V,E)$ with known disk centers}  \\
		\textbf{Output: }{A TRDF $f=(V_0,V_1,V_2)$ and the corresponding weight $W(f)$}
		\begin{algorithmic}[1] 
			\State $V_0=\emptyset$, $V_1=\emptyset$, $V_2=\emptyset$, $V^{\prime}=V$
			
			\While{$V^{\prime}\neq \emptyset$} \Comment{\textit{Roman} property of TRDF}\label{alg2:begin_loop1_find_dom} 
		    \State choose a vertex $v\in V^{\prime}$ 
		    \State $V_2=V_2\cup \{v\}$ \label{assign2}
		    \State $V^{\prime}=V^{\prime}\setminus N_G[v]$ 
		    \EndWhile   \label{alg2:end_loop1_find_dom}

			\State $i=1$, $S=\emptyset$
			\For{each $u_i\in V\setminus V_2$}  \label{alg2:begin_loop2_find_neighbor} 
 		    \State $S_i=N_G(u_i)\cap V_2$
		    \State $i=i+1$, $S=\emptyset$
                \State $S=S\cup S_i$ 
		    \EndFor \label{alg2:end_loop2_find_neighbor}
                \State $S'=GreedySetCover(V_2,S)$\label{alg2:V1}
                \State $V_1=\{u_i|S_i\in S'\}$  \Comment{\textit{total} property of TRDF} \label{total}
		    \State $V_0=V\setminus (V_1\cup V_2)$ \label{alg2:V0}
		    \State \Return $f=(V_0,V_1,V_2)$ and  $W(f)=2\times |V_2|+|V_1|$           \label{alg2:return_f}
			\caption{\textit{TRDF-UDG-SC(G)}}
			\label{alg2:TRDF_UDG_SC}
		\end{algorithmic}
	\end{algorithm}
% 	----------------------------------------------------------------

\begin{lemma}\label{lem:correctness2}
The function $f=(V_0,V_1,V_2)$ in \textit{TRDF-UDG-SC} is a TRDF on $G$.
\end{lemma}
\begin{proof}
 \textit{TRDF-UDG-SC} runs in two phases. In the first phase, it finds a maximal independent set $V_2$ of $G$ (since every maximal independent set is a dominating set) and then assigns Roman value $2$ to each vertex in $V_2$  (see line~ \ref{alg2:begin_loop1_find_dom}-\ref{alg2:end_loop1_find_dom} of Algorithm~\ref{alg2:TRDF_UDG_SC}), which ensures the \textit{Roman property} of TRDF. To ensure the \textit{total property} of TRDF, it uses $GreedySetCover(V_2,S)$ to find another set $V_1$  by adding a neighbor vertex for each vertex in $V_2$ (see line~\ref{alg2:V1}-\ref{total} of Algorithm~\ref{alg2:TRDF_UDG_SC}). The remaining vertices in $G$ carry Roman value $0$ (see line~\ref{alg2:V0} of Algorithm~\ref{alg2:TRDF_UDG_SC}). Therefore, combinedly the nominated points in $V_2$ and $V_1$ satisfies the \textit{Roman} and \textit{total} properties.  Hence, the function $f=(V_0,V_1,V_2)$ is a TRDF on $G$.
\end{proof}
\begin{lemma}\label{lem:TimeComplexity2}
  \textit{TRDF-UDG-SC} runs in $O(n\log{k})$ time, where $n$ is the input size and $k$ is the size of the set with Roman value $2$.  
\end{lemma}
\begin{proof}
    The proof is similar to Lemma~\ref{lem:TimeComplexity1}.
\end{proof}
%----------------- sub-section 5.2 ------------------
\subsection{Analysis}
Let $D^*$ be an optimal dominating set, $f^*$  be an optimal TRDF of the given UDG $G$,  and $C^*$ be an optimal set cover of the set cover instance $<V_2,S>$, where $W(f^*)$ is the weight associated with $f^*$. Since $V_2$ is a maximal independent set, from Lemma~\ref{lem:sub-5}, we have 
\begin{equation}\label{eq-4}
    |V_2|\leq \frac{44}{9}|D^*|
\end{equation}
From Theorem~\ref{th:DomRom}, we conclude that 
\begin{equation} \label{eq-5}
    2|D^*|\leq W(f^*)
\end{equation} 
Since $f=(V_0,V_1,V_2)$ is a TRDF on $G$ (see Lemma~\ref{lem:correctness2}), the associated weight $W(f)$ of the TRDF $f$ is given by:
\begin{equation} \label{eq-6}
\begin{split}
W(f)& =2\times |V_2|+|V_1|\\
& \leq 2\times \frac{44}{9}|D^*|+H(5)\times |C^*| \text{ (refer to Equation~\ref{eq-4} and Theorem~\ref{th:cormen})}\\
& \leq \frac{44}{9}\times 2|D^*|+\frac{137}{60}\times |D^*|  \text{ (refer to Lemma~\ref{lem:D_VminusD} )}\\
& \leq \frac{44}{9}\times W(f^*)+\frac{137}{120}\times 2|D^*| \\
& \leq \frac{44}{9}\times W(f^*)+\frac{137}{120}\times W(f^*)\text{ (refer to Equation~\ref{eq-5})}\\
& =\frac{2171}{360}\times W(f^*)\\ 
&\approx 6.03 \times  W(f^*) 
\end{split}
\end{equation}
\begin{theorem}
 The proposed algorithm (\textit{TRDF-UDG-SC}) gives a $6.03\operatorname{-}$ factor approximation result for the TRDF problem in UDGs. The algorithm runs in $O(n\log{k})$ time, where $n$ is the input size and $k$ is the size of the set with Roman value $2$ (i.e., $V_2$).  
\end{theorem}
\begin{proof}
    The approximation factor and the time complexity result follow from Equation~\ref{eq-6} and Lemma~\ref{lem:TimeComplexity2}, respectively.
\end{proof}
%--------------------- Section 6 ------------------------------------------
\section{Conclusion}\label{sec:conclusion}
In this paper, we studied the total Roman dominating set problem in unit disk graphs and proved that the problem is NP-complete. We proposed a $7.17\operatorname{-}$ factor and a $6.03\operatorname{-}$ factor approximation algorithm for TDS and  TRDS 
 problem in UDGs, respectively. The running time of both the algorithms is $O(n\log{k})$, where $n$ is the number of vertices in the input UDG and $k$ is the size of the set $D$ and $V_2$ in TDS and TRDS problem, respectively.

% \section{Remarks}   % or Conclusion  %  or Discussion  %  or just leave it out
% Please avoid changing anything in this template that will cause the fonts and margins to look different.\\
% You may use up to 6 pages, not including references or the appendix.
% Use pdflatex.\\

%---------------------------- Bibliography -------------------------------

% Please add the contents of the .bbl file that you generate,  or add bibitem entries manually if you like.
% The entries should be in alphabetical order
\small
\bibliographystyle{abbrv}
\bibliography{Paper3.bib}
% \begin{thebibliography}{99}

% % \bibitem{so2005}
% % C. So and H. So.
% % \newblock A groundbreaking result.
% % \newblock {\em Journal of Everything}, 59(2):23--37, 2005.

% \end{thebibliography}

% \newpage
% \section*{Appendix}
% Submissions should not exceed six pages (excluding references), must be submitted electronically, and must be prepared using LaTeX; using this template. Authors who feel that additional details are necessary should include a clearly marked appendix, which will be read at the discretion of the Program Committee. Each submission will be refereed by at least three members of the Program Committee. Details on the submission procedure are outlined on the conference website. Six-page papers accepted at CCCG will appear in the electronic proceedings of the conference, provided that they are presented by a speaker at the conference.
\end{document}